\documentclass[english,11pt]{article}
\usepackage[english]{babel}
\usepackage{latexsym,amssymb,amsmath,amsfonts,amscd,epsfig,amsthm,color}
\usepackage{graphicx}
\usepackage{float}
\usepackage[left=1in,right=1in,top=1in, bottom=1in]{geometry}
\usepackage{ifthen}
\usepackage{authblk}
\usepackage{colonequals}
\usepackage{array}

\newboolean{ElectronicVersion}
\setboolean{ElectronicVersion}{true}

\ifthenelse{\boolean{ElectronicVersion}}{
    \usepackage[letterpaper=true,pdftex,bookmarks,pagebackref,
	plainpages=false, 
        pdfpagelabels=true 
        ]{hyperref}}{}

\makeatletter
\newtheorem*{rep@theorem}{\rep@title}
\newcommand{\newreptheorem}[2]{%
\newenvironment{rep#1}[1]{%
 \def\rep@title{#2 \ref*{##1}}%
 \begin{rep@theorem}}%
 {\end{rep@theorem}}}
\makeatother

\newtheorem{theorem}{Theorem}[section]
\newtheorem{lemma}[theorem]{Lemma}
\newtheorem{corollary}[theorem]{Corollary}
\newtheorem{proposition}[theorem]{Proposition}

\newreptheorem{theorem}{Theorem}
\newreptheorem{lemma}{Lemma}

\def\ket#1{{\lvert}#1\rangle}
\def\field{\mathbb{F}_2}
\def\bra#1{{\langle}#1\rvert}

\def\ceil#1{{\lceil}#1\rceil}
\def\floor#1{{\lfloor}#1\rfloor}

\DeclareMathOperator{\Tr}{Tr}

\renewcommand{\(}{\left(}
\renewcommand{\)}{\right)}

\newcommand{\comproblem}[5]{
\noindent\begin{minipage}{\textwidth}
\vspace{2pt}
\hrule
\vspace{2pt}
\noindent\textbf{#1}, #2
\vspace{-7pt}
\begin{quotation}
\noindent\text{Alice's input:} #3\\
\text{Bob's input:} #4\\
\text{Output:} #5
\end{quotation}
\vspace{-5pt}
\hrule
\vspace{2pt}
\end{minipage}
}

\newcommand{\compromproblem}[6]{
\noindent\begin{minipage}{\textwidth}
\vspace{2pt}
\hrule
\vspace{2pt}
\noindent\textbf{#1}, #2
\vspace{-7pt}
\begin{quotation}
\noindent\text{Alice's input:} #3\\
\text{Bob's input:} #4\\
\text{Promise:} #5\\
\text{Output:} #6
\end{quotation}
\vspace{-5pt}
\hrule
\vspace{2pt}
\end{minipage}
}

\title{Quantum Communication Complexity of Distributed Set Joins}

\author[1]{Stacey Jeffery}
\author[2]{Fran\c{c}ois Le Gall}
\affil[1]{Institute for Quantum Information and Matter, California Institute of Technology, Pasadena, USA\\
  \texttt{sjeffery@caltech.edu}}
\affil[2]{Graduate School of Informatics, Kyoto University, Kyoto, Japan\\
  \texttt{legall@i.kyoto-u.ac.jp}}

\begin{document}

\maketitle

\begin{abstract}
Computing set joins of two inputs is a common task in database theory. Recently, Van Gucht, Williams, Woodruff and Zhang [PODS 2015] considered the complexity of such problems in the natural model of (classical) two-party communication complexity and obtained tight bounds for the complexity of several important distributed set joins.

In this paper we initiate the study of the \emph{quantum} communication complexity of distributed set joins. We design a quantum protocol for distributed Boolean matrix multiplication, which corresponds to computing the \emph{composition join} of two databases, showing that the product of two $n\times n$ Boolean matrices, each owned by one of two respective parties, can be computed with $\widetilde{O}(\sqrt{n}\ell^{3/4})$ qubits of communication, where $\ell$ denotes the number of non-zero entries of the product. Since Van Gucht et al.\ showed that the classical communication complexity of this problem is $\widetilde{\Theta}(n\sqrt{\ell})$, our quantum algorithm outperforms classical protocols whenever the output matrix is sparse. We also show a quantum lower bound and a matching classical upper bound on the communication complexity of distributed matrix multiplication over $\mathbb{F}_2$.

Besides their applications to database theory, the communication complexity of set joins is interesting due to its connections to direct product theorems in communication complexity. In this work we also introduce a notion of \emph{all-pairs} product theorem, and relate this notion to standard direct product theorems in communication complexity.
\end{abstract}

\section{Introduction}
\paragraph*{Background}
Set joins are basic operations in relational database theory. The notion of set join was introduced to the database community more than forty years ago by Codd \cite{Codd70} to express operations combining  two tables in relational databases. This seminal paper considered, in particular, the \emph{composition join}: given two (relational) databases $A$ and~$B$, $A$ represented as a subset of $\{1,\ldots,m\}\times \{1,\ldots,n\}$ and $B$ as a subset of $\{1,\ldots,n\}\times \{1,\ldots,m\}$, the composition join of $A$ and $B$ is the set $\{(i,j)\:|\:\exists k : (i,k)\in A \textrm{ and } (k,j)\in B\}\subseteq \{1,\ldots,m\}\times \{1,\ldots,m\}$. 
Many other join operations have been defined so far and have found many applications (see, e.g., \cite{Arasu+06,Codd70,Helmer+97,Leinders+JCSS07,Mamoulis03,Melnik+03,Ramasamy+00,VanGucht+PODS15}).
 
The computational complexity of join operations is naturally an important issue. Very recently Van Gucht, Williams, Woodruff and Zhang \cite{VanGucht+PODS15} have investigated this question in the two-party communication complexity model where one party owns the first database, the second party owns the second database, and both parties collaborate to compute the join of these two databases using as little communication as possible. This model is interesting for two main reasons. First, it models the natural and practical task of distributed computation of join operations. Second, in the communication complexity setting it is possible to show strong lower bounds on the complexity of problems. Indeed, one of the main contributions of~\cite{VanGucht+PODS15} was to show quantitative differences between the (communication) complexities of various join operations.

Many join operations studied in database theory actually correspond to fundamental and well-studied computational tasks in other areas of computer science. The composition join mentioned above, in particular, corresponds to Boolean matrix multiplication, one central problems in theoretical computer science: if we represent the database $A$ by an $m\times n$ matrix $M_A$ and $B$ by an $n\times m$ matrix $M_B$ (such that $M_A[i,j]=1$ if and only if $(i,j)\in A$, and similarly for $M_B$), the matrix representation of the composition join of $A$ and $B$ is precisely the output of the Boolean matrix multiplication of $M_A$ and $M_B$ ({i.e.}, the $m\times m$ matrix~$C$ such that $C[i,j]=\bigvee_{k=1}^n M_A[i,k]\land M_B[k,j]$). The result by Van Gucht et al.\ on the communication complexity of the composition join \cite{VanGucht+PODS15} shows that the communication complexity of Boolean matrix multiplication is $\widetilde \Theta(n\sqrt{\ell})$ for the square case $m=n$ (a more complicated formula is also given for the rectangular case), where~$\ell$ denotes the number of non-zero entries in the product $C$. Since the parameter~$\ell$ represents the sparsity of the output matrix, algorithms and communication protocols with complexity depending explicitly on~$\ell$ are sometimes called output-sensitive and have been studied in several settings other than communication complexity as well \cite{Amossen+09,Buhrman+SODA06,JKLM15,Lingas11}.

\paragraph*{Our Results}

In this paper we initiate the study of the \emph{quantum} communication complexity of distributed set joins. Our main result is about the set joins related to matrix multiplication. We first show that the quantum communication complexity of the composition join ({i.e.}, Boolean matrix multiplication) is $O(\sqrt{n}\ell^{3/4}\log m)$ (Theorem \ref{thm:bmm-protocol}). This is better than the best possible classical protocol, which costs $\Omega(n\sqrt{\ell})$ as mentioned above. We also consider matrix multiplication over the binary field and show that its quantum communication complexity (and actually even its classical communication complexity) is $\widetilde O(n\sqrt{\ell})$ (Theorem \ref{theorem:upperbounds-fields}). We give a matching lower bound as well (Theorem \ref{th:MM-lower}). 

These bounds are also interesting since they confirm and substantiate our current understanding of the power of quantum algorithms for problems related to matrix multiplication. Indeed, while matrix multiplication over a field seems harder than Boolean matrix multiplication for quantum computers, we currently do not have any technique to prove such a statement in the time complexity setting. Our results prove this statement in the communication complexity setting, for instances with sparse output matrices. 

In addition to these concrete results, this work presents several interesting new open problems. An \emph{OR lemma} is a composition lemma that says that the quantum communication complexity of the function $f^{\vee m}(a_1,\dots,a_m,b_1,\dots,b_m)=\bigvee_{i=1}^m f(a_i,b_i)$ is at least $\Omega(\sqrt{m}Q(f))$, where $Q(f)$ is the quantum communication complexity of $f$. We show that our upper bound for composition join is tight up to logarithmic factors assuming the problem of Boolean matrix multiplication satisfies an OR lemma (Proposition \ref{prop:bmm-or-lemma}). We give further evidence that our upper bound is indeed tight by showing that it is tight at extreme values of $\ell$, when $\ell=O(1)$ (Proposition \ref{prop:bmm-lower-bound1}) and when $\ell=\Omega(n^2)$ (Proposition~\ref{prop:bmm-lower-bound2}). 

We believe that proving lower bounds on set joins is a very interesting area of future research, as doing so may give insight into direct product theorems in communication complexity, as well as lower bounds in quantum query complexity for problems that involve \emph{read-many formulas}, in which different parts of the input are used multiple times, which makes it difficult to prove lower bounds using standard composition theorems.

\paragraph*{Organization}

The remainder of this paper is organized as follows. In Section \ref{sec:prelim}, we give the necessary preliminaries, including quantum communication complexity, and the groundwork for studying the quantum communication complexity of set joins. In Section \ref{sec:BMM}, we present our communication protocol for composition join. In Section \ref{sec:MM}, we present our classical communication protocol and matching quantum lower bound for matrix multiplication over $\mathbb{F}_2$. Finally, in Section~\ref{sec:lower}, we give some evidence that our upper bound for composition join is tight, by reducing a matching lower bound to a plausible OR lemma.

\section{Preliminaries}\label{sec:prelim}
\subsection{Notation}
Let $2^{[n]}$ denote the set of subsets of $[n]$. A subset $S$ of $[n]:=\{1,\dots,n\}$ can be represented by an $n$-bit string, and we will sometimes conflate these two notions. Let $S[i]$ denote the $i$-th bit of the string corresponding to $S$, so $S[i]=1$ if and only if $i\in S$. For any $x\in\{0,1\}^n$, we let $|x|$ denote the Hamming weight, which is the size of the corresponding subset of $[n]$. Similarly, for a Boolean matrix $A$ ({i.e.}, a matrix with entries in $\{0,1\}$), let $|A|$ denote the number of 1s in $A$.  Given an $m\times n$ Boolean matrix $A$ and an $n\times m$ Boolean matrix $B$, we write the Boolean product $A\ast B$ and let $AB$ denote their product over the finite field $\field=\{0,1\}$. Morever, for any integer $k\in[n]$ we let $A[\cdot,k]$ denote the $k$-th column of $A$ and $B[k,\cdot]$ the $k$-th row of $B$.

\subsection{Quantum Communication Complexity}

A communication problem is a function $f:A\times B\rightarrow Y$ whose input has two parts, $a\in A$, which we call Alice's input, and $b\in B$, which we call Bob's input. In the model of communication complexity, first defined by Yao \cite{Yao79}, Alice and Bob want to run a protocol such that, at the end of the protocol, Alice and Bob both output $f(a,b)$ with high probability, and they want to minimize the number of bits they need to communicate in order to achieve this. 

In the model of quantum communication, also introduced by Yao \cite{Yao93}, Alice and Bob are allowed a quantum communication channel and they want to minimize the number of quantum bits (qubits) they need to communicate in order to compute the function. More precisely, a quantum communication protocol consists of finite inner product spaces ${\cal X}$ and ${\cal Y}$, a measurement $\{\Pi,I-\Pi\}$ on $\cal Y$, and unitary operators $\{U_i\}_{i=1}^T$, such that for odd $i$, $U_i$ acts on ${\cal X}\otimes {\cal C}$, where ${\cal C}=\mathbb{C}^2$ is a single-qubit system, and for even $i$, $U_i$ acts on $\mathcal{C}\otimes {\cal Y}$. The protocol is said to have quantum communication complexity $T$. The protocol is said to compute $f$ with bounded error $1/3$ if for all $(a,b)\in A\times B$, there exist states $\rho_a$ and $\rho_b$ on ${\cal X}$ and ${\cal Y}$ respectively, such that 
\begin{equation*}
|\Tr\((I_{\cal X}\otimes I_{\mathcal{C}}\otimes \Pi)U_T\dots U_1(\rho_a\otimes \ket{0}\bra{0}\otimes \rho_b)\)  -f(a,b)|\leq 1/3.
\end{equation*}
That is, Alice begins the protocol in some state $\rho_a$ depending on her input, and Bob begins the protocol in some state $\rho_b$  depending on his input, and Alice also has an additional system, ${\cal C}$, initialized to $\ket{0}$, which will be used for communication with Bob. Alice applies $U_1$ to ${\cal X}\otimes \mathcal{C}$, and then sends ${\cal C}$ to Bob, who applies $U_2$ to ${\cal C}\otimes {\cal Y}$, before sending ${\cal C}$ back to Alice. They continue until they have applied all $T$ unitaries, at which point, Bob measures ${\cal Y}$, and the outcome determines $f(a,b)$ with error at most $1/3$. 

The \emph{bounded error quantum communication complexity} of $f$, denoted $Q(f)$, is the minimum~$T$ such that there exists a quantum communication protocol computing $f$ with bounded error $1/3$ with quantum communication complexity $T$. We will also consider the bounded error quantum communication complexity of partial functions $f:D\rightarrow \{0,1\}$ for $D\subseteq A\times B$.

There are many variants of this model, including the setting of \emph{one-way} communication complexity, in which Alice can send messages to Bob, but Bob cannot send messages to Alice, and only Bob is required to output the correct answer. We let $Q^1(f)$ denote the one-way communication complexity of $f$. 

An important problem in the study of quantum communication complexity is the problem of set disjointness, which is defined as follows.\vspace{2mm}

\comproblem{Set Disjointness}{$\mathsf{DISJ}_n$}
{$a\in\{0,1\}^n$}
{$b\in \{0,1\}^n$}
{$\mathsf{DISJ}_n(a,b)=\bigvee_{i=1}^na_ib_i$}

It is well known that $Q({\sf DISJ}_n)=\Theta(\sqrt n)$ \cite{BCW98,HW02,AA03,Raz03}, beating the classical communication complexity of $\Theta(n)$ \cite{KS92,Raz92}. When one of the two input sets is small, we can do even better as shown in the following elementary lemma. 

\begin{lemma}[Set disjointness for small sets]\label{lem:disj}
The bounded error quantum communication complexity of $\mathsf{DISJ}_n(a,b)$ is $O\(\sqrt{\frac{\min\{|a|,|b|\}}{|a\cap b|+1}}\log n\)$. Furthermore, if $\mathsf{DISJ}_n(a,b)=1$, then the protocol also returns a uniform random $i\in a\cap b$. 
\end{lemma}
\begin{proof}
To begin the protocol, Alice sends Bob $|a|$ using $\ceil{\log_2 n}$ bits of communication, and Bob sends Alice $|b|$ using $\ceil{\log_2 n}$ bits of communication. If $|a|<|b|$, Alice and Bob perform Grover search on the set $S_A=\{i\in [n]:a_i=1\}$ for an index $i\in S_A$ such that $f_B(i)=1$, where $f_B(i)=b_i$. They do this as follows. Alice computes $\ket{\pi(S_A)}=\sum_{i\in S_A}\frac{1}{\sqrt{|a|}}\ket{i}$. In order to perform the search, Alice and Bob must alternate $R_A=2\ket{\pi(S_A)}\bra{\pi(S_A)}-I$ and $R_B=\sum_{i\in [n]}(-1)^{b_i}\ket{i}\bra{i}$, $O\(\sqrt{\frac{|a|}{|a\cap b|+1}}\)$ times. Clearly Alice can implement $R_A$, and Bob can implement $R_B$, so they can implement this algorithm using $O\(\sqrt{\frac{|a|}{|a\cap b|+1}}\)$ rounds of communication, each time communicating a $\ceil{\log_2 n}$-qubit state. Bob measures some $i\in [n]$, and sends $i,f_B(i)$ to Alice. Both Alice and Bob output $f_B(i)$. If $|a|\geq |b|$, they do the protocol obtained by reversing Alice and Bob's roles. 
\end{proof}

The algorithm in Lemma \ref{lem:disj} actually finds a witness $i\in a\cap b$, which is slightly stronger than what is required to solve $\sf DISJ$. We will also consider the problem of finding the entire intersection:\vspace{2mm}

\comproblem{Find-all Set Intersection}{$\mathsf{DISJall}_n$}
{$a\in\{0,1\}^n$}
{$b\in \{0,1\}^n$}
{$\mathsf{DISJall}_n(a,b)=\{i\in [n]:a_i=b_i=1\}$}

In this case, we also have an advantage when $a$ or $b$ is small, as shown in the following lemma.

\begin{lemma}[Find-all set intersection for small sets]\label{lem:disj-find-all}
The bounded error quantum communication complexity of $\mathsf{DISJall}_n(a,b)$ is $O(\sqrt{|a\cap b|\min\{|a|,|b|\}}\log n)$.
\end{lemma}
\begin{proof}
Alice and Bob run the following protocol.

\begin{enumerate}
\item $S\leftarrow \emptyset$, $\tilde{a}\leftarrow a$, $\tilde{b}\leftarrow b$.
\item Repeat:
	\begin{enumerate}
	\item Use the protocol for $\mathsf{DISJ}_n(\tilde{a},\tilde{b})$ to obtain $i\in \tilde{a}\cap \tilde{b}$. If $\mathsf{DISJ}_n(\tilde{a},\tilde{b})=0$, output $S$.  
	\item $S\leftarrow S\cup \{i\}$, $\tilde{a}_i\leftarrow 0$, $\tilde{b}_i\leftarrow 0$.
	\end{enumerate}
\end{enumerate}

This protocol has communication complexity 
\begin{equation*}
\sum_{i=1}^{|a\cap b|}\sqrt{\frac{\min\{|a|,|b|\}}{|a\cap b|-i+1}}\log n=\Theta\(\sqrt{|a\cap b|\min\{|a|,|b|\}}\log n\)
\end{equation*}
qubits.\qedhere
\end{proof}

\subsection{Set Joins and Direct Product Theorems}

In this paper, we consider various \emph{set join} problems. For any predicate ${\cal P}_n:2^{[n]}\times 2^{[n]}\rightarrow \{0,1\}$, we can define a set join, as follows. \vspace{2mm}

\comproblem{$\cal P$-Set Join}{${\cal P}_{n}^{\otimes m}$}{${\cal A}=(A_1,\dots,A_m)$, $A_i\subseteq [n]$}{${\cal B}=(B_1,\dots,B_m)$, $B_i\subseteq [n]$}{$\{(i,j)\in [m]\times [m]:{\cal P}_n(A_i,B_j)=1\}$}

When ${\cal P}_n$ is the predicate such that ${\cal P}_n(A,B)=1$ if and only if $A\cap B\neq \emptyset$, the resulting join is called the \emph{composition join} or sometimes \emph{set-intersection-non-empty join}. As mentioned in the introduction, this join is equivalent to Boolean matrix multiplication, where we consider $A_1,\dots,A_m$ to be the rows of a matrix $A\in\{0,1\}^{m\times n}$, and $B_1,\dots,B_m$ to be the columns of a matrix $B\in \{0,1\}^{n\times m}$.

Consider a related construction: the direct product. \vspace{2mm}

\comproblem{Direct product}{${\cal P}_n^{(m)}$}{${\cal A}=(A_1,\dots,A_m)$, $A_i\subseteq[n]$}{${\cal B}=(B_1,\dots,B_m)$, $B_i\subseteq [n]$}{$\{i\in [m]:{\cal P}_n(A_i,B_i)\}$}

Unlike set joins, such problems are well-studied, and much is known. Clearly, we have $Q({\cal P}_n^{(m)})={O}(mQ({\cal P}_n)\log m)$ for any predicate ${\cal P}_n$. Intuitively, one can usually expect that the resources needed to solve $m$ instances of ${\cal P}_n$ scale as at least $m$ times the resources needed to solve one instance, that is: $Q({\cal P}_n^{(m)})=\Omega(mQ({\cal P}_n))$. This is called a \emph{(weak) direct product theorem} for ${\cal P}_n$. In fact, we can sometimes prove a stronger statement: that even solving ${\cal P}_n^{(m)}$ with success probability $2^{-m}$ requires $\Omega(mQ({\cal P}_n))$ quantum communication. Such a statement is called a \emph{strong direct product theorem}. Although such a statement likely holds for many problems in quantum communication complexity, it can be very difficult to prove (see, e.g., \cite{SherstovSICOMP12} and the references therein). 

In the case of set joins, it is also easy to see that $Q({\cal P}_n^{\otimes m})={O}(m^2Q({\cal P}_n)\log m)$, however, unlike the case of direct products, this naive upper bound is often not tight. For example, let $Q^1({\cal P}_n)$ denote the \emph{one-way communication complexity} of ${\cal P}_n$. Then we have the following:

\begin{theorem}
For any predicate ${\cal P}_n$, $Q({\cal P}_{n}^{\otimes m})\leq O(mQ^1({\cal P}_n)\log m)$.
\end{theorem}
\begin{proof}
Consider an optimal one-way quantum communication protocol for ${\cal P}_n$. Let $\rho(A)$ be the mixed state on at most $Q^1({\cal P}_n)$ qubits that Alice sends Bob and let $U(B)$ be the unitary that Bob applies to $\rho(A)\otimes\ket{0}\bra{0}_{\cal W}\otimes \ket{0}\bra{0}_{\cal A}$, for some workspace $\cal W$ and single-qubit answer register $\cal A$, so that he measures ${\cal P}_n(A,B)$ in the answer register with probability at least $2/3$.

We construct a (one-way) protocol for ${\cal P}_n^{\otimes m}$ as follows. Let Alice have input $A_1,\dots,A_m$, and Bob $B_1,\dots,B_m$. For every $i\in [m]$, Alice sends Bob $(\rho(A_i))^{\otimes c\log m}$, where $c$ is a large enough constant. For each $i,j\in [m]$, Bob applies $U(B_j)^{\otimes c\log m}$ to $(\rho(A_i)\otimes \ket{0}\bra{0}_{\cal W}\otimes \ket{0}\bra{0}_{\cal A})^{\otimes c\log m}$. He then computes the majority of the answer registers in a new single-qubit register, which he measures. 
Let $\rho(A_i,B_j):=U(B_j)(\rho(A_i)\otimes \ket{0}\bra{0}_{\cal W}\otimes\ket{0}\bra{0}_{\cal A})U(B_j)^\dagger=\sum_{b,b'\in\{0,1\}}\rho_{b,b'}\otimes \ket{b}\bra{b'}$, so the state Bob measures is (up to permuting registers):
\begin{equation*}
\sum_{x,x'\in \{0,1\}^\ell}\bigotimes_{i=1}^\ell \rho_{x_i,x_i'}\otimes \ket{x}\bra{x'}\otimes \ket{\text{maj}(x)}\bra{\text{maj}(x')},
\end{equation*}
where $\text{maj}(x)=1$ if $|x|\geq \ell/2$ and $0$ otherwise. Assume that ${\cal P}_n(A_i,B_j)=1$, as the $0$ case is nearly identical. Then the probability of success in a single round is $\Tr(\rho_{1,1})\geq 2/3$, so the probability of success upon measuring the majority register is: 
\begin{equation*}
\sum_{x\in\{0,1\}^\ell:|x|\geq \ell/2}\Pi_{i=1}^\ell\Tr(\rho_{x_i,x_i})=\sum_{k=0}^{\floor{\ell/2}}\binom{\ell}{k}\Tr(\rho_{1,1})^k(1-\Tr(\rho_{1,1}))^{\ell-k}\geq 1-e^{-\Omega(\ell)}=1-m^{-\Omega(1)},
\end{equation*}
where the inequality follows from Hoeffding's inequality, and the constant in $\Omega(1)$ depends on $c$. Thus, Bob gets the correct answer with high probability, but furthermore, this measurement causes negligible damage to the state $\rho(A_i,B_j)^{\otimes \ell}$, so Bob can apply $(U(B_j)^\dagger)^{\otimes \ell}$ to recover $\rho(A_i)^{\otimes \ell}$, to be used again. 
The error in the state remains negligible as long as Bob does this no more than $m^{O(1)}$ times.
\end{proof}

Call a theorem of the form $Q({\cal P}_n^{\otimes m})=\Omega(\min\{mn,m^2Q({\cal P}_n)\})$ a \emph{(weak) all-pairs product theorem}. The $\min\{mn,\cdot\}$ is to account for the fact that we always have a trivial upper bound of $mn$, and so if we did not include this, the statement would always be false for some values of $m$ and $n$. In this work, we give an example of a set-join for which an all-pairs direct product theorem does not hold --- in particular, in Section \ref{sec:BMM} we will give an upper bound of $O(m^{3/2}\sqrt{n})$ for the composition join, showing that this problem does not satisfy an all-pairs product theorem. Although we show that such a statement holds for matrix multiplication over $\mathbb{F}_2$, in that case, we have $\min\{mn,m^2Q({\cal P}_n)\}=mn$ for all $m$ and $n$, so the best strategy is always for Alice to send her whole input to Bob, rather than for Alice and Bob to compute $m^2$ instances of ${\cal P}_n$. It is an open question whether or not there exists a predicate for which an all-pairs product theorem holds in a non-trivial sense --- that is, the best strategy is to compute $m^2$ instances of ${\cal P}_n$.

\section{Composition Join (Boolean Matrix Multiplication)}\label{sec:BMM}

In this section, we give an upper bound on the communication complexity of Boolean matrix multiplication (equivalent to computing the composition join), proving our main theorem. As in \cite{VanGucht+PODS15}, we consider the following promise version of the problem, in which the output has at most $\ell$ ones.\vspace{2mm}

\compromproblem{Boolean Matrix Multiplication}{$\mathsf{BMM}_{m,n,\ell}$}{$A\in\{0,1\}^{m\times n}$}{$B\in \{0,1\}^{n\times m}$}{$|A\ast B|\leq \ell$}{$\mathsf{BMM}_{m,n,\ell}(A,B)=A\ast B=\{(i,j)\in [m]\!\times\! [m]:\exists k\in [n], A[i,k]=B[k,j]=1\}$}

The communication protocol we give is inspired by the query-optimal quantum algorithm for Boolean matrix multiplication given in \cite{JKLM15}. 
The algorithm of \cite{JKLM15} is based on a subroutine for a problem called \emph{graph collision}. For any family of bipartite graphs $G$ on $n$ vertices, the communication version of graph collision on $G$ is as follows.\vspace{2mm}

\comproblem{Graph Collision}{$\mathsf{GC}_G$}{$f_A\in\{0,1\}^n$}{$f_B\in\{0,1\}^n$}{$\mathsf{GC}_G(f_A,f_B)=\bigvee_{(i,j)\in G}f_A(i)f_B(j)$}

An efficient protocol for this problem can easily be constructed in the communication complexity setting:
\begin{lemma}[Graph collision]\label{cor:gc}
$Q(\mathsf{GC}_G(f_A,f_B))=O(\sqrt{\min\{|f_A|,|f_B|\}})$ for any family of bipartite graphs $G$.
\end{lemma}
\begin{proof}
Alice sends Bob $|f_A|$, and Bob sends Alice $|f_B|$ using $2\log n$ bits of communication. If $|f_A|\leq |f_B|$, Alice sets $a=f_A$ and Bob sets $b=\{i\in [n]:\exists j\in [n], (i,j)\in G, f_B(j)=1\}$. Otherwise, Alice sets $a=\{j\in [n]: \exists i\in [n], (i,j)\in G, f_A(i)=1\}$ and Bob sets $b=f_B$. They finally compute $\mathsf{DISJ}(a,b)$. 
\end{proof}

When we solve graph collision as a subroutine, we will actually want to additionally find \emph{all} graph collisions in a particular instance. That is, we will want to solve the following problem. \vspace{2mm}

\comproblem{Find All Graph Collisions}{$\mathsf{GCall}_G$}
{$f_A\in\{0,1\}^n$}
{$f_B\in \{0,1\}^n$}
{$\mathsf{GCall}_G(f_A,f_B)=\{(i,j)\in G:f_A(i)=f_B(j)=1\}$}

\noindent The following upper bound for ${\sf GCall}_G$ is a corollary of the previous lemma (its proof is similar to the proof of Lemma \ref{lem:disj-find-all}).

\begin{corollary}[Find all graph collisions]\label{cor:gc-all}
$Q(\mathsf{GCall}_G(f_A,f_B))=O(\sqrt{\lambda\min\{|f_A|,|f_B|\}})$, where $\lambda = |\{(i,j)\in G:f_A(i)=f_B(j)=1\}|$.
\end{corollary}

The final ingredient we need before presenting our quantum communication protocol for Boolean matrix multiplication is a quantum communication protocol that searches for a 1-instance among $n$ independent instances of a communication problem. Its proof is fairly straightforward and simply combines quantum search with the original communication protocol.

\begin{lemma}[Search over communication instances]\label{lem:comm-search}
Let $f:X\times Y\rightarrow \{0,1\}$ be a communication problem with bounded error quantum communication complexity $Q(f)$. Let $F:X^n\times Y^n\rightarrow \{0,1\}$ be the problem of finding some $i\in [n]$ such that $f(x_i,y_i)=1$. Then $Q(F)=O(\sqrt{\frac{n}{t}}Q(f)\log n)$, where $t=|\{i\in [n]:f(x_i,y_i)=1\}|$. 
\end{lemma}
\begin{proof}
Alice creates $\ket{\pi}=\sum_{i\in [n]}\frac{1}{\sqrt{n}}\ket{i}$. Alice and Bob implement quantum search by repeating the reflections 
\begin{equation*}
R_1=2\ket{\pi}\bra{\pi}-I\quad\mbox{ and }\quad R_2=\sum_{i\in [n]}(-1)^{f(x_i,y_i)}\ket{i}\bra{i}
\end{equation*}
$O(\sqrt{n/t})$ times. Each implementation of $R_2$ is accomplished as follows. Let Alice's state be $\sum_{i\in [n]}\alpha_i\ket{i}$. Alice performs the mapping $\ket{i}\mapsto \ket{i,i}$, to get $\sum_{i\in [n]}\alpha_i\ket{i,i}$ and sends half of the state to Bob. Conditioned on their quantum state, Alice and Bob perform the protocol for $f$ using input $(x_i,y_i)$, that is, they perform the protocol on a superposition of inputs. This leaves the state $\sum_{i \in [n]}\alpha_i\ket{i,f(x_i,y_i)}_A\ket{i,f(x_i,y_i)}_B$ (here we assume, without loss of generality, that the final state in the protocol for $f$ does not contain any garbage). Alice then maps this state to $\sum_i (-1)^{f(x_i,y_i)}\ket{i,f(x_i,y_i)}_A\ket{i,f(x_i,y_i)}_B$. They run the protocol in reverse to uncompute $f(x_i,y_i)$, leaving $\sum_i\alpha_i(-1)^{f(x_i,y_i)}\ket{i}_A\ket{i}_B$.
Bob sends his half to Alice, so she can uncompute it, leaving the state $\sum_i\alpha_i(-1)^{f(x_i,y_i)}\ket{i}$, and thus implementing~$R_2$.  
\end{proof}

\noindent We are now ready to state and prove our main theorem.
\begin{theorem}[Upper bound for Boolean matrix multiplication]\label{thm:bmm-protocol}
For all $\ell\in\{1,\dots,m^2\}$, $$Q({\sf BMM}_{m,n,\ell})=O(\sqrt{n}\ell^{3/4}\log m).$$
\end{theorem}
\begin{proof}
Alice and Bob run the following communication protocol. \vspace{1mm}

\begin{enumerate}
\item Alice and Bob individually store the all-zero matrix $C$ of size $m\times m$.
\item Repeat:
	\begin{enumerate}
	\item Alice and Bob jointly find $k\in [n]$ such that $\mathsf{GC}_{\overline{C}}(A[\cdot,k],B[k,\cdot])=1$. If none exists, Alice and Bob output $C$.\label{item:find}
	\item Alice and Bob jointly compute $S\leftarrow \mathsf{GCall}_{\overline{C}}(A[\cdot,k],B[k,\cdot])$.\label{item:find-all}
	\item Alice and Bob individually compute $C\leftarrow C+S$.
	\end{enumerate}
\end{enumerate}

In this protocol Alice and Bob each maintain a matrix $C$ containing the 1s of the product $A\ast B$ found by the protocol so far.  They repeatedly search for a new $k\in [n]$ such that $f_A^k=A[\cdot,k]$ and $f_B^k=B[k,\cdot]$ have graph collisions with respect to the graph given by the complement of $C$. When they find such a $k$, they compute all graph collisions. Suppose they find $\{k_1,\dots,k_t\}$ before there are no more $k$ to be found, and let $\lambda_{i}$ be the number of ones found at round $i$. By Lemma \ref{cor:gc} and Lemma \ref{lem:comm-search}, in round $i$ step \ref{item:find} costs $O\(\sqrt{\frac{n}{t-i+1}\min\{|f_A^{k_i}|,f_B^{k_i}|\}}\log m\)$. By Corollary \ref{cor:gc-all}, in round $i$ step \ref{item:find-all} costs $O\(\sqrt{\lambda_i\min\{|f_A^{k_i}|,|f_B^{k_i}|\}}\log m\)$. 
Thus, the total cost is at most:
\begin{equation*}
\sum_{i=1}^t\(\sqrt{\frac{n}{t-i+1}}\sqrt{\min\{|f^{k_i}_A|,|f^{k_i}_B|\}}+\sqrt{\lambda_i\min\{|f_A^{k_i}|,f_B^{k_i}|\}}\)\log m.
\end{equation*}
Note that for any $k$, every $(i,j)$ such that $f_A^k(i)=f_B^k(j)=1$ implies that $(A\ast B)[i,j]=1$, so we necessarily have $\ell\geq |f_A^k|\cdot|f_B^k|$. We therefore have $\min\{|f_A^k|,|f_B^k|\}\leq \sqrt{\ell}$ for all $k\in [n]$, and thus, the total cost is at most (up to constants):
\begin{eqnarray*}
\ell^{1/4}\sum_{i=1}^t\(\sqrt{\frac{n}{t-i+1}}+\sqrt{\lambda_i}\)\log m & \leq & \ell^{1/4}\(\sqrt{nt}+\sqrt{t\sum_{i=1}^t\lambda_i}\)\log m\\
&\leq & \ell^{1/4}(\sqrt{nt}+\sqrt{t\ell})\log m,
\end{eqnarray*}
where in the first line we use the fact that $\sum_{i=1}^ti^{-1/2}=\Theta(\sqrt{t})$ and Cauchy-Schwartz inequality, and in the second line we use the fact that $\sum_{i=1}^t\lambda_i=\ell$, since $\ell$ is the total number of ones we find over all rounds. Finally, observe that since $t$ is the number of distinct witnesses $k\in [n]$ found, $t\leq n$, and since we find at least one new 1 in every round except the last, we also have $t\leq \ell+1$. Thus, the total communication is at most 
\begin{equation*}
(\ell^{1/4}\sqrt{n}+\ell^{3/4})\sqrt{\min\{n,\ell\}}\log m=O\( \ell^{3/4}\sqrt{n}\log m\),
\end{equation*}
as claimed.\qedhere
\end{proof}

\section{Matrix Multiplication over Finite Fields}\label{sec:MM}

In this section we consider matrix multiplication over finite fields and give tight bounds (up to possible polylogarithmic factors) on its communication complexity. We work out here only the case of square matrices over the binary field. Formally, the problem we consider is the following. \vspace{2mm}

\compromproblem{Square matrix multiplication over $\mathbb{F}_2$}{$\mathsf{MM}_{n,\ell}$}{$A\in\mathbb{F}_2^{n\times n}$}
{$B\in\mathbb{F}_2^{n\times n}$}
{$|AB|\le \ell$}
{the matrix $AB\in\mathbb{F}_2^{n\times n}$}

The main result of this section is the following upper bound on the classical (and thus quantum) communication complexity of this problem.
\begin{theorem}[Upper bound for matrix multiplication over $\field$]\label{theorem:upperbounds-fields}
The classical communication complexity of $\mathsf{MM}_{n,\ell}$ is $\widetilde O(n\sqrt{\ell})$.
\end{theorem}
We will need two lemmas to prove Theorem~\ref{theorem:upperbounds-fields}.
The first lemma is a finite-field version of a result related to compressed sensing used in~\cite{VanGucht+PODS15}.  
The proof of this finite-field version can be found in \cite{Draper+ISIT09}.
\begin{lemma}\label{lemma:compressed-sensing}
For any positive integer $n$ and any integer $\kappa\in\{1,\ldots,n\}$, there are a distribution on random matrices
$M\in \field^{O(\kappa)\times n}$ and a reconstruction function $\mathsf{Rec}(\cdot)$ such that for any vector $x\in\field^n$ with at most $\kappa$ non-zero entries the inequality 
\[
\Pr_M\big[\mathsf{Rec}(Mx)=x\big]>0.99.
\]
holds ({i.e.}, $\mathsf{Rec}(\cdot)$ applied on $Mx$ returns $x$ with high probability).
\end{lemma}

The second lemma shows how to use Freivalds' technique to detect non-zero columns of a matrix product.
Similar ideas were used in \cite{Gasieniec+ISAAC14}.
\begin{lemma}\label{lemma:Freivalds}
Let $m$ and $n$ be two positive integers.
Consider the setting where Alice has for input a matrix $A\in\field^{m\times n}$ and Bob has for input a matrix $B\in\field^{n\times n}$. Alice and Bob can detect, with high probability, which columns of $AB$ contain at least one non-zero entry with $\widetilde O(n)$ communication.
\end{lemma}
\begin{proof}
Consider the following procedure: Alice takes a vector $v$ uniformly at random in $\field^m$; Alice sends the row-vector $v^TA\in\field^{n}$ to Bob; Bob sends the row-vector $v^TAB\in\field^{n}$ to Alice. This procedure has communication complexity $2n$ and, for each column of $AB$, enables Alice and Bob to decide with probability at least 1/2 whether this column contains at least one non-zero entry. By repeating this procedure a logarithmic number of times, Alice and Bob are able to find, with high probability, which columns of $AB$ contain at least one non-zero entry. 
\end{proof}

\noindent We are now ready to prove Theorem \ref{theorem:upperbounds-fields}.
\begin{proof}[Proof of Theorem \ref{theorem:upperbounds-fields}]
We assume for convenience that both $\sqrt{\ell}$ and $n/\sqrt{\ell}$ are integers (the general case is handled similarly).
We will say that a column of $AB$ is dense if it contains at least $0.9\sqrt{\ell}$ non-zero entries, and say that a column of $AB$ is sparse if it contains at most $1.1\sqrt{\ell}$ non-zero entries (note that a column can be both sparse and dense).
The protocol is as follows.
\begin{itemize}
\item[1.]
Alice and Bob partition the columns of $AB$ into dense columns and sparse columns: they compute a set of indexes $S\subseteq\{1,\ldots,n\}$ such that, for any $j\in\{1,\ldots,n\}$, the $j$-th column of $AB$ is dense if $j\in S$ and sparse if $j\notin S$.
\item[2.]
Alice and Bob compute all entries of all columns of $AB$ with index in $S$.
\item[3.]
Alice and Bob compute all entries of all the columns of $AB$ with index in $[n]\setminus S$.
\end{itemize}

Step 1 can be done probabilistically with $\widetilde O(n)$ bits of communication by repeating the following procedure: Alice constructs a $(n/\sqrt{\ell})\times n$ matrix $A'$ by selecting $n/\sqrt{\ell}$ rows of $A$ uniformly at random; Alice and Bob then use the protocol of Lemma \ref{lemma:Freivalds} (with $A'$ as Alice's input and $B$ as Bob's input) to decide which columns of $A'B$ have more than one non-zero entry. Repeating this procedure a logarithmic number of times enables Alice and Bob to decide, with high probability, which columns of $AB$ are not dense: for a non-dense column of $AB$ ({i.e.}, a column with less than $0.9\sqrt{\ell}$ non-zero entries) the corresponding column of $A'B$ will not contain any non-zero entry with high probability (on the choice of~$A'$). The indices of the other columns are collected in $S$. The indices in $S$ thus correspond only to dense columns of $AB$. While the set $S$ may not contain the indices of all the dense columns of $AB$, it can be seen from a similar argument that all non-sparse columns of $AB$ ({i.e.}, the columns with at least $1.1\sqrt{\ell}$ non-zero entries) will be put in $S$, which means that
all indices in $[n]\setminus S$ correspond to columns of $AB$ that are sparse.

Step 2 can be done with $O(|S|n)=O(\sqrt{\ell}n)$ bits of communication (note that $|S|\le \frac{1}{0.9}\sqrt{\ell}$ since $AB$ has only at most $\ell$ non-zero entries): Bob simply sends the entries $B[i,j]$ for all $(i,j)\in \{1,\ldots,n\}\times S$, and then Alice computes $AB[i,j]$ for all $(i,j)\in \{1,\ldots,n\}\times S$.

Step 3 can be done with $\widetilde O(n\sqrt{\ell})$ bits of communication using Lemma~\ref{lemma:compressed-sensing} with $\kappa=\ceil{1.1\sqrt{\ell}}$, by repeating the following procedure a logarithmic number of times: Alice chooses a random matrix~$M$ as in Lemma~\ref{lemma:compressed-sensing} and sends $MA$ to Bob; for each $j\in \{1,\ldots,n\}\setminus S$, Bob computes $\mathsf{Rec}(MAz)$ where~$z$ denotes the $j$-th column of $B$. 
\end{proof}

We now show a lower bound on the quantum (and thus also classical) communication complexity of matrix multiplication over $\field$, which matches the upper bound of Theorem \ref{theorem:upperbounds-fields} up to  polylogarithmic factors.

\begin{theorem}\label{th:MM-lower}
The quantum communication complexity of $\mathsf{MM}_{n,\ell}$ is $\Omega(n\sqrt{\ell})$.
\end{theorem}
\begin{proof}
Assume for convenience that $\sqrt{\ell}$ is an integer (the general case is handled similarly).
Let $x^1,\ldots,x^{\sqrt{\ell}},y^1,\ldots,y^{\sqrt{\ell}}$ be $2{\sqrt{\ell}}$ vectors in $\mathbb{F}_2^n$. 
Let $x\in\mathbb{F}_2^{n\sqrt{\ell}}$ be the vector obtained by concatenating $x^1,\ldots,x^{\sqrt{\ell}}$,
and $y\in\mathbb{F}_2^{n\sqrt{\ell}}$ be the vector obtained by concatenating $y^1,\ldots,y^{\sqrt{\ell}}$.

Construct the $n\times n$ matrix $A$ by putting the vector $x^i$ as its $i$-th row, for each $i\in\{1,\ldots,\sqrt{\ell}\}$, and setting the next $n-\sqrt{\ell}$ rows to zero (observe that $\sqrt{\ell}\le n$ since $\ell\le n^2$).
Construct the $n\times n$ matrix $B$ by putting the vector $y^j$ as its $j$-th column, for each $j\in\{1,\ldots,\sqrt{\ell}\}$, and setting the next $n-\sqrt{\ell}$ columns to zero.
Observe that $|AB|\le\ell$ and the parity of the diagonal entries of the matrix product $AB$ is equal to 
\[
\bigoplus_{i=1}^{\sqrt{\ell}} x^i\cdot y^i=x\cdot y.
\]
We thus obtain a reduction from computing the inner product of two vectors in $\mathbb{F}_2^{n\sqrt{\ell}}$ to solving $\mathsf{MM}_{n,\ell}$. Since the quantum communication complexity of the former problem is $\Omega(n\sqrt{\ell})$, as shown in \cite{Kremer95}, we obtain the same lower bound for $\mathsf{MM}_{n,\ell}$.
\end{proof}

\section{Lower Bounds for Boolean Matrix Multiplication}\label{sec:lower}

An important open problem of this work is to prove a tight lower bound on the bounded error quantum communication complexity of Boolean matrix multiplication, {i.e.}, to show that the upper bound of Theorem \ref{thm:bmm-protocol} is tight. Let us focus on the square case ({i.e.}, $m=n$). We are able to prove two lower bounds, each of which is tight for one extreme value of $\ell$: $\ell=O(1)$ or $\ell=\Omega(n^2)$, but neither is tight for the range $\ell\in (\omega(1),o(n^2))$. We further show that assuming a plausible OR-lemma, our upper bound is indeed tight, up to logarithmic factors. 

\begin{proposition}\label{prop:bmm-lower-bound1}
For all $\ell\in\{1,\dots,n^2\}$, $Q({\sf BMM}_{n,n,\ell})=\Omega(\sqrt{n\ell})$.
In particular, when $\ell=O(1)$, then $Q(\mathsf{BMM}_{n,n,\ell})=\Omega(\sqrt{n}\ell^{3/4})$. 
\end{proposition}
\begin{proof}
We can embed $\sqrt{\ell}\leq n$ instances $\{(a^{(i)},b^{(i)})\}_{i=1}^{\sqrt{\ell}}$ of $\mathsf{DISJ}_n$ in an instance of $\mathsf{BMM}_{n,n,\ell}$ as follows. Let $A$ have $a^{(i)}$ in the $i$-th row for $i=1,\dots,\sqrt{\ell}$, and all zeros elsewhere, and let $B$ have $b^{(i)}$ in the $i$-th column for $i=1,\dots,\sqrt{\ell}$, and zeros elsewhere. Then $AB$ is 0 except in the upper left $\sqrt{\ell}\times\sqrt{\ell}$ submatrix, so $|AB|\leq \ell$, and $(AB)[i,i]={\sf DISJ}_n(a^{(i)},b^{(i)})$, so $AB$ encodes ${\sf DISJ}_n^{(\sqrt{\ell})}(a^{(1)},\dots,a^{(\sqrt{\ell})},b^{(1)},\dots,b^{(\sqrt{\ell})})$. The result follows from the $\Omega(\sqrt{\ell}\sqrt{n})$ lower bound on $Q({\sf DISJ}_n^{(\sqrt{\ell})})$ shown in \cite{KSW07}. 
\end{proof}

\begin{proposition}\label{prop:bmm-lower-bound2}
For all $\ell\in\{1,\dots,n^2\}$, $Q({\sf BMM}_{n,n,\ell})=\Omega(\ell)$. In particular, when $\ell=\Omega(n^2)$, then $Q(\mathsf{BMM}_{n,n,\ell})=\Omega(\ell^{3/4}\sqrt{n})$.
\end{proposition}
\begin{proof}
We can embed an instance $(a,b)$ of the inner product function $\mathsf{IP}_{\ell}$ in an instance of $\mathsf{BMM}_{n,n,\ell}$ as follows. Let $B=I$ be the identity matrix, and let $A$ contain the $\ell$-bit string $a$ in the first $\ell$ positions. Then Alice and Bob jointly compute $AB=A$, and Bob can compute $\mathsf{IP}(a,b)$ and send the resulting bit to Alice. Since $Q(\mathsf{IP}_{\ell})=\Omega(\ell)$, we have $Q(\mathsf{BMM}_{n,n,\ell})=\Omega(\ell)$. 
\end{proof}

\begin{proposition}\label{prop:bmm-or-lemma}
Suppose computing the entrywise-OR of $k$ independent instances of $\mathsf{BMM}_{n,n,n^2}$ has bounded error quantum communication complexity $\Omega(\sqrt{k}Q(\mathsf{BMM}_{n,n,n^2}))$. Then for any $\ell\in [n^2]$, $Q(\mathsf{BMM}_{n,n,\ell})=\Omega(\ell^{3/4}\sqrt{n})$.
\end{proposition}
\begin{proof}
Let $(A_1,B_1),\dots, (A_k,B_k)$ be independent instances of $\mathsf{BMM}_{\sqrt{\ell},\sqrt{\ell},\ell}$, for $k=\frac{n}{\sqrt{\ell}}$. Define $A$ and $B$ as follows:
$$A=\left[\begin{array}{ccc}
A_1 & \dots & A_k\\
0 & \dots & 0\\
\vdots & \ddots & \vdots \\
0 & \dots & 0 \end{array}\right],\qquad B=\left[\begin{array}{cccc}
B_1 & 0 & \dots & 0\\
\vdots & \vdots & \ddots & \vdots \\
B_k & 0 & \dots & 0 \end{array}\right].$$
Then $A\ast B$ has $\bigvee_{i=1}^k A_i\ast B_i$ in the top left corner, and zeros elsewhere. So $|A\ast B|\leq \ell$, and computing $A\ast B$ costs at least $\sqrt{k}Q(\mathsf{BMM}_{\sqrt{\ell},\sqrt{\ell},\ell})\geq \sqrt{\frac{n}{\sqrt{\ell}}}\ell=\ell^{3/4}\sqrt{n}$. 
\end{proof}

The above proposition actually holds equally true for the non-promise problem ${\sf BMM}_{m,n}={\sf BMM}_{m,n,m^2}$, and would imply $Q({\sf BMM}_{m,n})=\Omega(m^{3/2}\sqrt{n})$. 

\section*{Acknowledgments}
The authors are grateful to Troy Lee and Ronald de Wolf for helpful discussions. FLG is supported by the Grant-in-Aid for Young Scientists~(A)~No.~16H05853 and the Grant-in-Aid for Scientific Research~(A)~No.~16H01705 of the Japan Society for the Promotion of Science, and the Grant-in-Aid for Scientific Research on Innovative Areas~No.~24106009 of the Ministry of Education, Culture, Sports, Science and Technology in Japan. SJ is supported by the Institute for Quantum Information and Matter, 
an NSF Physics Frontiers Center (NFS Grant PHY-1125565) with support of the Gordon and Betty Moore Foundation (GBMF-12500028).

\small
\bibliographystyle{alpha}
\bibliography{bmm}

\end{document}